\documentclass[a4paper,12pt]{article}
\usepackage{amssymb,amsmath,amsthm}
\usepackage{newpxtext} 
\usepackage{newpxmath} 
\usepackage{graphicx}
\usepackage[dvipsnames]{xcolor}
\usepackage{dsfont} 
\usepackage{tikz}
\usetikzlibrary{patterns}
\usepackage{etoolbox} 

\usepackage{natbib}
\setcitestyle{aysep={}} 

\usepackage{soul}
\sethlcolor{Goldenrod!50}

\makeatletter
\patchcmd{\@maketitle}{\LARGE}{\fontsize{16}{21}\bfseries\selectfont}{}{}
\makeatother

\usepackage[indentafter]{titlesec}
\titleformat{name=\section}{}{\large\thetitle.}{0.8em}{\large}
\titleformat{name=\subsection}[runin]{}{\thetitle.}{0.5em}{\bfseries}[.]
\titleformat{name=\subsubsection}[runin]{}{\thetitle.}{0.5em}{\itshape}[.]
\titleformat{name=\paragraph,numberless}[runin]{}{}{0em}{}[.]
\titlespacing{\paragraph}{0em}{0em}{0.5em}
\titleformat{name=\subparagraph,numberless}[runin]{}{}{0em}{}[.]
\titlespacing{\subparagraph}{0em}{0em}{0.5em}

\usepackage[left=1in,top=1in,right=1in,bottom=1in,headheight=0.8in,foot=0.5in]{geometry}
\usepackage[nodisplayskipstretch]{setspace}
\usepackage{enumitem}
\setlist[enumerate]{itemsep=0mm}

\usepackage{hyperref}

\hypersetup{
	colorlinks=true,
	linkcolor=MidnightBlue,
	citecolor=MidnightBlue,
	urlcolor=MidnightBlue
      }

\newcommand{\RR}{\mathds{R}} 
\newcommand{\CC}{\mathds{C}} 

\newcommand{\hodge}{{\star}}

\DeclareMathOperator{\Imag}{Im}

\let\Re\relax 
\DeclareMathOperator{\Re}{Re}
\newcommand{\noopsort}[1]{} 

\usepackage{thmtools} 
\usepackage{thm-restate}
\theoremstyle{definition}
\newtheorem*{thegaugecriterion}{The Gauge Criterion}

\theoremstyle{plain}

\title{Gauge Matters More Than Curvature}
\date{\today}
\usepackage{authblk}
\author[1]{Henrique Gomes\thanks{gomes.ha@gmail.com}}
\affil[1]{\emph{University of Oxford, University of Bonn}}
\author[2]{Bryan W Roberts\thanks{b.w.roberts@lse.ac.uk}}
\affil[2]{\emph{London School of Economics \& Political Science}}

\begin{document}
\setstretch{1.2}
\maketitle

\begin{abstract}
Recent orthodoxy, which we call \textit{Curvature Fundamentalism}, posits that the field strength $F$ on spacetime exhausts the physically significant content of electromagnetism, at least on prosaic (contractible) spacetimes. We argue that it is gauge invariance that constrains physical significance. As a result, there are physically significant quantities that $F$ does not capture, even on contractible spacetimes. We show that this result is robust: in the presence of charged matter fields, there are interference experiments whose outcome cannot be predicted by $F$ or any other point-local gauge invariant quantities on the basis of initial Cauchy data. We prove this using the example of relative phase of a Klein-Gordon-Maxwell field, and use it to identify a context in which the widely-claimed equivalence between the Maxwell-Faraday formulation $EM_1$ and the gauge formulation $EM_2$ of electromagnetism is untenable, and that new strategies may be needed for the reductionist project.
\end{abstract}


\section{Introduction}
\setstretch{1.4}

What does it mean to say that the world is filled with physical fields, such as electromagnetic fields, Higgs fields, and the like? In electromagnetism, one view amongst philosophers is that there is a preferred field on Minkowski spacetime that determines all the physically significant content of the theory: the `field strength' or `Maxwell-Faraday' or `curvature' field $F$. We refer to this view as \textit{Curvature Fundamentalism}. It has consequences for other aspects of the interpretation of electromagnetism, such as in the literature on theoretical equivalence between scientific theories. For example, \citet[p.1078]{weatherall2016ng} assumes it when he writes that ``the empirical content of a model is exhausted by its associated Faraday tensor'', and uses this to argue that, on Minkowski spacetime, Maxwell's formulation of classical electromagnetism in terms of fields $F$ and $J$ (which he calls $EM_1$) is theoretically equivalent to the gauge formulation in terms of fibre bundle structures (which he calls $EM_2$). Similar usage of Curvature Fundamentalism appears in the broader literature on the philosophy of symmetry in Yang-Mills gauge theories.

We will argue that one should begin with a different postulate: a quantity in gauge electromagnetism is physically significant only if it is gauge invariant. Curvature is a candidate for a physically significant quantity because it is gauge invariant. But, it is not the only such candidate. Our argument begins with the observation that, historically and conceptually, gauge electromagnetism involves \emph{charged matter fields} defined on an associated vector bundle. Given such fields, curvature $F$ and four-current $J$ cannot capture all the physically significant facts about gauge electromagnetism---not even on a contractible manifold like Minkowski spacetime. The quantity that we will use to illustrate this is called \emph{relative phase}. It is widely discussed in physics, but has received limited attention from philosophers. We will use it to illustrate an interference experiment consisting of fields evolving dynamically from an initial data surface to a region of interference, and which is completely invisible to curvature. Indeed, we will exhibit outcomes that cannot be predicted using any `point-local' gauge invariant quantities on the basis of initial Cauchy data.\footnote{By a \emph{point-local quantity} we mean a quantity whose value at a spacetime point $p$ is determined by a finite number of derivatives of fields, which is to say the $k$-jet of the fields at $p$ for some finite $k$.\label{fn:point-local}} The same example shows that the widely claimed equivalence of $EM_1$ and $EM_2$ is not tenable in the presence of charged matter fields, and that alternative strategies may be needed for the reductionist.\footnote{Note that we do \emph{not} deny the mathematical validity of the theorems claiming to establish this equivalence, such as those of \citet[Proposition 5.5]{weatherall2016ng} and \citet[Proposition 2]{chen2024s}. We rather argue that the conclusion of these theorems has more limited scope than has generally been recognised in the philosophical literature, and that they make an awkward comparison between formulations of electromagnetism; see Section \ref{subs:curvature-fundamentalism} for details.}

In summary, we defend three claims, that in a complete formulation of classical electromagnetism with charged matter, the following hold:
\begin{itemize}
    \item \textit{Gauge constrains significance.} Gauge invariance is necessary for physical significance, although other factors like measurement and dynamics matter too.\footnote{There remains further debate as to whether one should \textit{begin} with gauge transformations in the stipulational sense of \citet{baker2023s}, or whether an alternative is preferable, such as the motivational account of \citet{moller-nielsen2017i} or the geometry-first account of \citet{gomes2026vb}. We will not comment on this issue here.}
    \item \textit{Curvature is not enough.} Gauge electromagnetism predicts physically significant behaviour that no finite-order point-local quantities like $(F,J)$ can predict on the basis of initial Cauchy data, even on contractible regions, and even when the four-current $J$ is constrained to be the Noether current of a dynamical field equation.
    \item \textit{The equivalence between $EM_1$ and $EM_2$ has restricted scope.} The Maxwell-Faraday formulation and the gauge formulation may be equivalent on contractible spacetimes when charged matter is fully captured by four-current, but not for general charged matter fields.
\end{itemize}

We begin in Section \ref{sec:gaugeEM} with a review of the historical and conceptual foundations of gauge electromagnetism, where we develop the essential role for charged matter. In Section \ref{sec:curvature-fundamentalism} we introduce and clarify the statement of Curvature Fundamentalism in the philosophical literature. Section \ref{sec:inequivalence} then turns to quantities that go beyond Curvature Fundamentalism, and introduces our main substantial example of relative phase. We argue that this example makes both Curvature Fundamentalism and the equivalence of $EM_1$ and $EM_2$ untenable in electromagnetism with charged matter.

\section{Gauge Electromagnetism}\label{sec:gaugeEM}

\subsection{The origin of gauge}\label{subs:origin-of-gauge} 

In \emph{A Treatise on Electricity and Magnetism}, Maxwell took the primary objects of his study to be the electric and magnetic fields. He then noticed that they could each be written in terms of a potential $A$, and that they were left unchanged by the transformation $A\mapsto A+d\chi$. This led him to conclude,\footnote{Maxwell's complete view on the status of the electromagnetic potential is subtle. In his early work `On Physical Lines of Force', \citet{maxwell1861plf-iia} referred to $A=(F,G,H)$ as the `electrotonic state', suggesting that it captured part of the true state of the electromagnetic system. In the later \emph{Treatise} he switched to calling $A$ the `potential', but still suggested that by removing the unphysical $d\chi$ term, one could determine the ``true values of the components of $A$'' \citep[p.235]{maxwell1873em-ii}.} ``[t]he quantity $\chi$... is not related to any physical phenomenon'' \citep[p.55]{maxwell1873em-ii}. Today we call this a \emph{gauge transformation}. In summary, Maxwell assumed that electric and magnetic fields are physically significant, and then used this fact to argue that gauge transformations are symmetries. In contrast, the gauge formulation of electromagnetism turns this picture around: we \emph{begin} with gauge transformations, and then define a physically significant quantity to be one that those transformations leave invariant. How did this remarkable inverted perspective emerge?

The term \emph{gauge} (\emph{Eichung}) was coined by \citet{weyl1918,weyl1919g}, who proposed a local metrical scaling symmetry that transforms its associated one-form as $A\mapsto A+d\phi$, as part of his attempt to unify general relativity and electromagnetism. This led \citet{kaluza1921z} and \citet{klein1926k} to their famous interpretation of gauge transformations as acting on a fifth dimension of spacetime.\footnote{See \citet[\S 5.3]{roberts2025con} for a brief history and philosophy of Kaluza-Klein theory.} Although \citet[p.967]{kaluza1921z} himself called this an ``extremely odd'' interpretation, all three authors adopted the influential  methodology of taking gauge transformations to determine physical significance, rather than the other way around. This was the birth of what we will call the `Gauge Criterion' for physical significance.

\citet{fock1926g} interpreted this new dimension as simply ``some new parameter with the dimension of the quantum of action'' \citep[p.839]{fock1926g}, rather than a higher spatial dimension. He proposed that the purpose of this new internal degree of freedom $\theta\in [0,2\pi)$ was just ``to ensure the invariance of the equations under the addition of an arbitrary gradient to the four-potential'' on the road to deriving a gauge invariant Klein-Gordon equation \citep[p.840]{fock1926g}. Thus, Fock adopted the view that gauge invariance constrains physical significance, but without the apparent scandal of introducing a new dimension of space. \citet{gordon1926c} and \citet{london1927w} did the same, as did \citet{weyl1929eug} himself when he came to review these developments.

Gauge invariance soon became the guide to physical significance for a vast number of different theories. \citet{yangmills1954g} \cite[and independently][]{shaw1954phd} showed that by replacing the abelian gauge group $G=U(1)$ of electromagnetism with an alternative (non-abelian) group $G=SU(2)$, one could model alternative phenomena like isospin. This laid the foundation for the theory of electroweak interactions, and eventually the Standard Model. \citet{utiyama1956g} generalised gauge groups to include any Lie group $G$, reportedly working independently of \citeauthor{yangmills1954g} \citep[pp.208--9]{oraifeartaigh1997}. Meanwhile, the mathematics of connections on fibre bundles were developed independently by \citet{ehresmann1952c} and his followers, and eventually summarised in a widely-studied textbook by \citet{KobayashiNomizu1963fdg}.

These developments culminated with the discovery that gauge theory appears to make novel physical predictions about electromagnetism. The turning point was the remarkable discovery of \citet{aharonovbohm1959a}, that in the configuration space of a test charge near an impenetrable solenoid, there is a quantity associated with measurable electromagnetic interference around a loop $\lambda$ that is given (in units of $\hbar=c=1$) by,
\begin{equation}\label{eq:ab-relative-phase}
  \theta_\lambda = e^{iq\oint_\lambda A}.
\end{equation}
This is a gauge-invariant quantity whose value can be varied by switching on the solenoid, all while maintaining fixed (vanishing) curvature $F$ in the region accessible to the charge. \citeauthor{aharonovbohm1959a} interpreted this as a measurable interference pattern that is not predicted by Maxwell's theory, and which was quickly verified experimentally by \citet{chambers1960ab}. This surprising development led \citet[p.3845]{wuyang1975g} to argue that the Maxwell-Faraday tensor $F$ ``by itself does not... completely describe all electromagnetic effects on the wavefunction of the electron.'' Instead, gauge invariant quantities like Equation \eqref{eq:ab-relative-phase} are needed as well.\footnote{\citeauthor{wuyang1975g} went on to provide what is now known as the `Wu-Yang dictionary' as a way of translating between physical gauge concepts and fibre bundle concepts. Isadore Singer received this manuscript from Yang while visiting Stony Brook and took it back to Oxford, where Michael Atiyah and other developers of the modern mathematical form of gauge theory were able to study it for the first time \citep[pp.274--5]{JiWang2025yang}. See also \citet[\S 1.3.2]{gomes2025el}.}

We would like to point out two lessons that arose out of the discovery of the gauge formulation of electromagnetism: 1) there is a bundle of `internal degrees of freedom' together with a structure group $G$ acting on each fibre, which describe the properties of potentials and charged matter fields; and 2) although these properties are not all physically significant, the gauge-invariant properties are candidates for physical significance. In this sense, the gauge formulation drops the privileged physical significance of $F$ in  Maxwell's theory. The result is a more complete description of electromagnetism and a more general theory of internal degrees of freedom. We will now briefly review its structure.

\subsection{Principal and Associated Bundles}\label{subs:gauge-formalism}

Yang-Mills gauge theories are defined on a principal $G$-bundle $\pi:P\to M$ for some Lie group $G$, where $G$ is called the \textit{gauge group}, and where the base space $M$ is often understood as a spacetime manifold endowed with the Minkowski metric $\eta$, but could more generally be a configuration space for charged matter. The definition of a principal $G$-bundle includes a free smooth action of $G$ on $P$ from the right, in such a way that the fibres $P_p := \pi^{-1}(p)$ of each point $p$ in $M$ are group orbits, and where the bundle is locally trivial.\footnote{A \textit{locally trivial bundle} with principal fibre $F$ is one which has the local structure of a Cartesian product of spacetime with $F$, in the sense that each point $p\in M$ admits a neighbourhood $U$ and a bundle isomorphism $\phi$ such that each fibre $P_q$ is mapped as $\phi:P_q\mapsto\{q\}\times F$ for all $q\in U$. } Thus, we may interpret the fibre over a point as representing the possible values of `internal structure' at that point, and the gauge group $G$ as smoothly cycling through those values. The principal $G$-bundle is then equipped with a \emph{principal connection} $\alpha$, which provides the notion of a gauge potential.

To represent charged matter, we couple a vector space $V$ to the principal bundle by what is known as the \emph{associated vector bundle}\footnote{Formally, $\rho:G\to GL(V)$ is a representation of $G$ amongst the linear transformations of $V$, and $E=P\times_\rho V$ is the quotient of $P\times V$ by the right action $(p,v)\mapsto(pg,\rho(g^{-1})v)$. Equivalently, sections of $E$ correspond to equivariant maps from $P$ to $V$.} $E := P\times_\rho V$ over $M$. The vector space $V$ may consist of scalars, tensors, or spinors, and may be either real or complex. A \emph{section} of a bundle over $M$ is a map from $M$ back to the bundle that assigns to each point a value in the fibre over it. 

Often there is no global section $s:M\to P$ such that $\pi\circ s$ is the identity. Thus, we often make use of a \textit{local section} $s:U\to P$ whose domain is a local neighbourhood $U\subseteq M$. This allows one to associate the connection $\alpha$ with a real-valued \textit{local potential} one-form $A$ on spacetime defined by $-iA := s^*\alpha$, where $s^*$ denotes the pullback and where we represent the Lie algebra as $\mathfrak{u}(1):=i\RR$. Similarly, a charged matter field is a section $\psi$ of the associated bundle $E$. Relative to a local trivialisation it is represented by a $V$-valued function on $U$, and it is this representative that appears in a typical field equation.

In a general Yang-Mills gauge theory, a physical system is described by at least one\footnote{More are used in the presence of multiple kinds of charged matter, but for simplicity we will stick to just one, as nothing in our argument depends on this.} associated vector bundle $E=P\times_\rho V$ equipped with a connection. Given a local section, this gives rise to a local potential and a charged matter field $(A,\psi)$ on the base manifold. Their dynamical evolution is usually given by the subset of pairs $(A,\psi)$ satisfying a set of field equations, and which are referred to as the \emph{dynamical} or \emph{on-shell} models. All the interactions in the Standard Model can be described using these general structures.

The specific case of \textit{gauge electromagnetism} is any Yang-Mills theory in which the gauge group is $G=U(1)$. Charged matter of charge $q$ takes values in the charge-$q$ representation $\rho_q:U(1)\to GL(V)$, for non-zero integer values of $q$. One could again adopt a variety of different fields, but for simplicity we will restrict attention to the complex scalar field $V=\CC$, for which this representation has the form $\rho_q(e^{i\beta}) = e^{iq\beta}$. For a given choice of local section or `gauge', a \emph{(local) model of gauge electromagnetism} is a pair $(A,\psi)$, where $A$ is a local potential and $\psi$ is a charged matter field on spacetime. On each local section, the connection induces a \textit{covariant derivative} on the associated bundle, whose local expression for a charge-$q$ scalar field is $D_A := d - iqA$. These structures can be shown\footnote{Let $(u,z)\sim\bigl(uh,\rho_q(h^{-1})z\bigr)$ be the equivalence relation defining
$P\times_{\rho_q}V$ as a quotient space. Then, for each smooth map $g:U\to U(1)$, a change of local section from $s \mapsto s'=sg^{-1}$ induces a transformation $\psi\mapsto \psi'=\rho_q(g)\psi$ of the charged matter field. Writing $g=e^{i\chi}$ and using our definition $\rho_q(e^{i\chi})=e^{iq\chi}$, we get $\psi'=e^{iq\chi}\psi$ as claimed. Similarly, let $-iA=s^*\alpha$. The same $s\mapsto s'=sg^{-1}$ gives $(s')^*\alpha = s^*\alpha - (dg)g^{-1}$. Using $g=e^{i\chi}$ then gives $(s')^*\alpha = s^*\alpha - i(d\chi)$. This in turn implies $-iA' = -iA - i(d\chi)$, and hence $A' = A+d\chi$. The transformation rule for $D_A$ immediately follows.} to transform under a gauge transformation as,
\begin{align}
    A \mapsto A + d\chi && \psi\mapsto e^{iq\chi}\psi && D_A\psi\mapsto e^{iq\chi}D_A\psi
\end{align}
for some real scalar field $\chi$.  Two models $(A,\psi)$ and $(A',\psi')$ are said to be members of the same \textit{gauge equivalence class} whenever they are related by a gauge transformation,\footnote{This notion of a model is close in spirit to what \citet[p.674]{nguyentehwells2020g} call the ``groupoid of gauge fields $\mathcal{C}_A$'', whose objects are gauge fields and whose morphisms are gauge transformations. However, we do not wish to take a position on what the appropriate \textit{category} is for interpreting gauge electromagnetism, which is the subject of debate in the theory equivalence literature following \citet{weatherall2016ug}.}  $A'=A+d\chi$ and $\psi'=e^{iq\chi}\psi$ for some $\chi$.

We can now state the definition of physical significance that arose out of our brief history of Yang-Mills gauge theory.
\begin{thegaugecriterion}
A quantity in a Yang-Mills theory is physically significant only if it is invariant under gauge transformations.
\end{thegaugecriterion}
We emphasise that this criterion states a necessary condition, not a sufficient one. Gauge invariance guarantees that a quantity is well defined on gauge equivalence classes of models. However, whether a given gauge-invariant quantity is physically significant may have further meaning that is only settled by the dynamics and by the existence of a measurement protocol. Thus, to avoid arbitrarily gerrymandered functionals of fields that are gauge invariant without playing any measurement or explanatory role, our main example in Section \ref{sec:inequivalence} will include an explicit dynamics and measurement protocol.

One important object of study in general Yang-Mills gauge theory is the \emph{curvature} $\mathcal{F}$ of a connection $\alpha$, defined by $\mathcal{F} := d\alpha + \tfrac{1}{2}[\alpha,\alpha]$. Its representation in a given section can be written\footnote{We slightly abuse notation by using the same symbol $d$ to represent the exterior derivative on $P$ (in the definition of $\mathcal{F}$) and on $M$ (in the definition of $F$). Since each manifold admits a unique exterior derivative, the meaning will be clear from context.} $F = dA + \tfrac{1}{2}[A,A]$. Another important object is the \emph{four-current} $J$, which is the Noether current associated with the dynamical field equations. These structures are taken\footnote{The definition of the Hodge star operator makes use of spacetime structure on the underlying base manifold: if $(M,g)$ is an oriented Lorentzian manifold of dimension $n$, then $\hodge$ is the unique linear map from $p$-forms to $n-p$ forms such that every pair of $p$-forms $\alpha,\beta$ satisfies $\alpha\wedge\hodge\beta = \langle\alpha,\beta\rangle\mathrm{vol}$, where $\mathrm{vol}$ is the volume form determined by $g$ and the chosen orientation.} to satisfy the Yang-Mills equation $J = \pm\hodge D_A \hodge F$, which describes how dynamical evolution is constrained by curvature.

In electromagnetism, for which the gauge group $G=U(1)$ is abelian, curvature reduces to the simple form,
\begin{equation}
  \mathcal{F}=d\alpha.
\end{equation}
This quantity can be defined in a neighbourhood of each spacetime point: for any local section $s$ with $A := i(s^*\alpha)$, the quantity $F := i(s^*\mathcal{F}) = dA$ is an antisymmetric rank-$(0,2)$ tensor field. It is gauge invariant:
\begin{equation}
  F\mapsto d(A+d\chi) = dA + d^2\chi = dA = F,
\end{equation}
where we use the fact that $d^2=0$. Conversely, on a contractible region $U\subseteq M$ of the base manifold, curvature can always be written as $F=dA$ for some potential $A$ that is unique up to a gauge transformation.\footnote{On a contractible region, the Poincar\'e lemma ensures that $dF=0$ implies $F=dA$ for some $A$. Moreover, if $F=dA=dA'$, then $d(A'-A)=0$ and so by another application of the Poincar\'e lemma we have that $A' - A = d\chi$ for some scalar field $\chi$, i.e. $A'$ and $A$ are related by a gauge transformation.} When $F=dA$, this curvature tensor is automatically closed ($dF=0$). This fact can be used to identify a model of the Maxwell-Faraday formulation: the equation $dF=0$, together with the Yang-Mills equation $J = \hodge d\hodge F$, together provide an invariant expression of Maxwell's equations.\footnote{For a two-form $F$ related to $E$ and $B$ in the standard way for a given reference frame, $dF=0$ is equivalent to the joint statement of Faraday's law ($\nabla\times E = -\partial B/\partial t$) and no magnetic monopoles ($\nabla\cdot B=0$). Similarly, $J = \hodge d\hodge F$ is equivalent to the joint statement of Gauss's law ($\nabla \cdot E = \rho$) and the Amp\`ere-Maxwell law ($\nabla\times B = J + \partial E/\partial t$).} For this reason, the representation of curvature $F$ on the base manifold $M$ in $EM_2$ can be identified with the Maxwell-Faraday (or `electromagnetic field strength') tensor in $EM_1$.

\subsection{Classical Charged Matter Fields}\label{subs:classical-charged-matter}

Before moving on, we would like to briefly comment on the nature of charged matter in the gauge formulation, which involves an associated bundle $E=P\times_\rho V$. The vector space $V$ is typically complex, and for that reason charged matter fields $\psi$ have often been mistaken for intrinsically quantum objects. We argue that this is not necessary: these fields $\psi$ can be interpreted as fully within the scope of classical electromagnetism.

The ubiquity of complex-valued matter fields in electromagnetism arises out of the $U(1)$ gauge group: the continuous transformation of a field by $U(1)$ requires a representation of rotations, which in turn requires at least two real parameters if it is to be non-trivial. In contrast, rotations can be equivalently described using just one complex parameter, which is both ecumenical and convenient. In either case, the smooth assignment of numbers (real or complex) to points on spacetime defines what is commonly known as a \emph{classical scalar field}, with no deep foundational significance arising from the use of one or the other.\footnote{Some early developers of quantum mechanics seem to have thought that $c$-numbers must be real-valued for the purposes of quantisation, but this is not a mathematical or a conceptual necessity for quantisation \citep[cf.][]{roberts2018observables}.} Although complex numbers are convenient for representing the gauge fields of electromagnetism, they do not by themselves require any further quantum mechanical or Hilbert space structure.

However, one could be forgiven for being confused about this: the founders of gauge theory like \citet{fock1926g}, \citet{london1927w}, \citet{weyl1929eug}, \citet{yangmills1954g}, and \citet{aharonovbohm1959a} all cite quantum theory as motivating the definition of the charged matter field $\psi$. However, early writers like \citet{london1927w} also interpreted the quantum state $\psi$ following Broglie's matter-wave proposal, with its defining feature being that it is a complex-valued classical field. In contrast, from a modern perspective, being complex-valued alone does not make something quantum.\footnote{As \citet[p.212]{wallace2009anti} writes: ``\emph{complex fields are just special cases of real fields.} After all, `complex fields' are just as valid as real fields in classical field theory: they are just ordered pairs of real fields (of twice the dimension); any complex-linear theory is also a real-linear theory. Indeed, from this perspective, complex Klein-Gordon theory is just a special case of real Klein-Gordon theory with an internal degree of freedom.''} A complex-valued field must still be \textit{quantised} in order to produce the structure of a functioning quantum field theory, which includes being given a statistical interpretation. Merely specifying a complex scalar, Weyl, or Dirac field does not provide any of that. Of course, the experimental verification of the Aharonov-Bohm effect makes use of ordinary quantum interference, which led \citet[p.1]{caprez2007macroscopic} to refer to it as ``an essentially quantum mechanical effect.''\footnote{\citet[p.4859]{shech2018ab} repeats this point, writing that ``quantum and classical theories make strikingly different predictions'' regarding the Aharonov-Bohm effect. We disagree: our view is rather that it is gauge electromagnetism ($EM_2$) and spacetime field electromagnetism ($EM_1$) that make strikingly different predictions.} However, its conceptual basis --- meaning, the derivation of the effect from first principles --- arises entirely from the classical fields of an associated bundle model in gauge theory. In our view, a more correct thing to say about the Aharonov-Bohm effect is that it is an essentially \emph{gauge theoretical} effect, in that it cannot be predicted using the tensorial quantities $(F,J)$ that constitute Maxwell's formulation of electromagnetism. A similar claim can be made about the relative phase experiment we describe in Section \ref{sec:inequivalence} below.

\citet[p.2013]{earman2019ab} refers to the use of complex fields in gauge electromagnetism as a ``bastardized theory in which a quantized electron is subjected [to] an external classical electromagnetic field''. That may be true of certain semiclassical models of the electron, although see \citet{dougherty2021non} for a response. But, in the Yang-Mills approach to gauge electromagnetism, our view is rosier: complex fields are part of a perfectly unadulterated classical model deriving from an associated bundle. For a given section, the field representatives $(A,\psi)$ have the property that $\psi$ may be complex-valued. But, until it is quantised, we are still doing classical electromagnetism. So, let us now turn to the philosophy literature that seeks to interpret these classical structures. 

\section{Fundamentalist interpretations}\label{sec:curvature-fundamentalism}

\subsection{Curvature fundamentalism}\label{subs:curvature-fundamentalism}
There is a piece of orthodoxy that has recently emerged in the philosophy of physics, which will be the main target of our critique. We refer to it as \emph{Curvature Fundamentalism:} the claim that the physics of a model of electromagnetism is reducible to its curvature $F$. This idea is generally qualified by the assumption that the base space $M$ is contractible, since otherwise the Aharonov-Bohm effect poses an immediate challenge: the quantity $\theta_\lambda$ in Equation \eqref{eq:ab-relative-phase} is not reducible to $F$. But, with this qualification in place, and given the Gauge Criterion that physically significant quantities are gauge invariant, a first pass at stating Curvature Fundamentalism is the following:
\begin{quote}
  ($CF_1$) \emph{Determinacy:} If a quantity is gauge invariant,
  it is determined by $F$.
\end{quote}
For example, \citet[p.542]{belot1998em} writes that ``$E$ and $B$ capture almost
all of the gauge-invariant content of electromagnetism'', with the qualification ``almost'' just included to cover the Aharonov-Bohm effect. Similarly, \citet[494]{dewar2019s} states that on a contractible manifold, ``the electromagnetic field [$F$] determines all gauge-invariant quantities''.

In the presence of a charged matter field $\psi$, this expression of Curvature Fundamentalism is false for an uninteresting reason: if a set of models of the form $(A,\psi)$ is not further constrained, then $F=dA$ at most determines $A$ (up to a gauge transformation), not $\psi$. As a result, there are simple gauge invariant quantities\footnote{For a complex-valued scalar field $\psi$, the quantity $|\psi|^2=\bar{\psi}\psi$ is both gauge invariant and not determined by $F$. This continues to be true when $\psi$ is required to satisfy field equations like the coupled Maxwell-Klein-Gordon equations. We discuss this kind of example in more detail in Section \ref{sec:inequivalence}.} like $|\psi|$ that are clearly not determined by $F$, violating Determinacy. But, the spirit of Curvature Fundamentalism is perhaps rather that the `physically significant' quantities are still deeply constrained by $F$. To feel the plausibility of this idea, recall that $F$ constrains $J$ via the Yang-Mills equation $J=\hodge d\hodge F$. This in turn constrains $\psi$, since $J$ is the Noether current for its field equations, and which themselves usually depend non-trivially on $A$. If $\psi$ were so constrained as to be \emph{determined} by $F$, then Curvature Fundamentalism really gets some grip. This determination might be further supported by constraints owing to the nature of physical measuring devices. So, it is a much less trivial question whether there are \emph{physically significant} quantities that are not determined by $F$. This leads to a somewhat more careful expression of Curvature Fundamentalism, again assuming the qualification of contractibility:
\begin{quote}
  ($CF_2$) \emph{Empirical Sufficiency}. If a quantity is physically significant, then it is determined by $F$.
\end{quote}

This expression of Curvature Fundamentalism is widespread in the philosophy of physics: \citet[p.1042]{weatherall2016ug} states that ``one may take the empirical content of electromagnetism to be fully exhausted by the electromagnetic field [$F$]'', and repeats a similar claim in several related articles.\footnote{For example, \citet[p.1078]{weatherall2016ng} writes ``on both formulations, the empirical content of a model is exhausted by its associated Faraday tensor'', and \citet[p.3]{weatherall2019tep} says ``ultimately, the empirical significance of a model of
  electromagnetism is entirely encoded in the electromagnetic field strength''.} Similarly, \citet[p.673--4]{nguyentehwells2020g} write that ``the empirical content of the theory can be expressed solely in terms of the gauge-invariant fields $F$''. More recently, \citet[p.14]{BradleyWeatherall2026s} write that ``the empirical significance of the vector potential always runs through its associated Faraday tensor''.

A third expression of Curvature Fundamentalism in the literature relates to the ontology of electromagnetism. Some have pointed out that, given the Empirical Parsimony principle that the ontology of a theory consists in its physically significant quantities, Empirical Sufficiency implies:
\begin{quote}
  ($CF_3$) \emph{Ontological Primacy}. If a quantity is in the ontology of electromagnetism, then it is determined by $F$.
\end{quote}
For example, \citet[p.542]{belot1998em} refers to the traditional interpretation of electromagnetism as stipulating that ``the ontology of electromagnetism consists of physically real electric and magnetic fields''. \citet{moller-nielsen2017i} argues for this statement on the basis of Determinacy and Parsimony, concluding that $F$ should be ``taken to represent the genuine material ontology of the theory''. \citet{jacobs2023mfb} refers to a similar view as `$F$-realism', although he does not endorse it. The statement of and relationship between these versions of Curvature Fundamentalism are summarised in Table \ref{tab:curvature-fundamentalism}.

\begin{table}[t]
\centering
\renewcommand{\arraystretch}{1.3}
\begin{tabular}{p{0.32\linewidth} p{0.60\linewidth}}
\textbf{Commitment} & \textbf{Statement} \\
  \hline
\emph{Gauge Criterion}
    & Physically significant $\Rightarrow$ Gauge Invariant. \\
($CF_1$) \emph{Determinacy}
    & Gauge invariant $\Rightarrow$ Determined by $F$. \\
($CF_2$) \emph{Empirical Sufficiency}
    & Physically significant $\Rightarrow$ Determined by $F$. \\
\emph{Empirical Parsimony}
    & Ontology $\Rightarrow$ Physically significant. \\
($CF_3$) \emph{Ontological Primacy}
    & Ontology $\Rightarrow$ Determined by $F$. \\
\hline\hline
\multicolumn{2}{l}{(\emph{Gauge Criterion}) $+$ (\emph{Determinacy}) $\Rightarrow$ (\emph{Empirical Sufficiency})} \\
\multicolumn{2}{l}{(\emph{Empirical Sufficiency}) $+$ (\emph{Empirical Parsimony}) $\Rightarrow$ (\emph{Ontological Primacy})} \\
\hline
\end{tabular}
\caption{Variations of Curvature Fundamentalism and their logical relationships.}
\label{tab:curvature-fundamentalism}
\end{table}

Curvature fundamentalism can indeed be motivated by a Maxwell-inspired ontology for electromagnetism: suppose, as Maxwell did, that all the theory's physically significant facts are captured by the electric and magnetic fields (as expressed by $F$) together with the four-current $J$. Since $F$ determines $J := \hodge d\hodge F$, this would imply that the ontology, and indeed every physically significant fact about electromagnetism, is determined by $F$.

However, sometimes philosophers go on to import this commitment into the gauge formulation of electromagnetism as well. For example, \citet[p.1041]{weatherall2016ug} only considers models of gauge electromagnetism in which charged matter is captured entirely by the four-current. In that case, since charged matter is determined by $J$, it is also determined by $F$. This is the basis for the argument of \citet{weatherall2016ng,weatherall2016ug} that $EM_1$ and $EM_2$ are equivalent, and for the argument of \citet{BradleyWeatherall2026s} against sophistication about symmetries.

We are not convinced. Since classical gauge electromagnetism also includes models with charged matter, this characterisation of its ontology is incomplete. Or, at best, it depends on a hidden assumption, that $F$ determines the behaviour of all classical charged matter fields satisfying a physical field equation (at least when $M$ is contractible). In the next section, we will show that this hidden assumption is not true, rendering all three formulations of Curvature Fundamentalism untenable. In this context, arguments for the equivalence of $EM_1$ and $EM_2$ come apart as well, as do arguments against sophistication about symmetries, since both are only plausible given Curvature Fundamentalism. 

That said, we agree that Curvature Fundamentalism is true in the special case of electromagnetism on contractible manifolds and in which charged matter either vanishes or is determined entirely by the four-current.\footnote{\citet{weatherall2016ng,weatherall2016ug} is explicit about adopting this restriction, characterising the models of gauge electromagnetism as having the form $(\RR^4,\eta,A)$ rather than a structure like $(M,\eta,A,\psi)$. As a result, we do not deny the validity of his Proposition 5.5, nor of related statements like Proposition 8 of \citet{dewar2019s} or Proposition 2 of \citet{chen2024s}. We rather find these results to be artificially limited in scope.} In this case, Determinacy holds: either the model does not include any associated bundle, or if it does then the field $\psi$ determined by a section is also determined entirely by $J = \hodge d \hodge F$. In this case, curvature $F$ determines $J$, which determines $\psi$. Since the manifold is assumed to be contractible, curvature also determines the local potential $A$ up to a gauge transformation. Thus, this is a case in which $F$ determines $A$ up to a gauge transformation and, by hypothesis, determines $\psi$ through $J$; hence it determines all gauge invariant quantities constructed from the model. As we have already seen, this fact of Determinacy together with the Gauge Criterion implies that all physically significant quantities are determined by $F$.

However, that argument does not apply to gauge electromagnetism in its full generality. On both the historical and conceptual development introduced above, gauge electromagnetism includes charged matter fields given by sections of an associated vector bundle, which are not generally determined by the four-current. This makes the situation in which $EM_1$ and $EM_2$ are equivalent a rather awkward comparison, even on contractible manifolds: it is a comparison between,
\begin{enumerate}
  \item Maxwell-Faraday models $(F,J)$, and
  \item Gauge models $(A,\psi)$ restricted to the sector in which $\psi$ either vanishes or is uniquely determined on-shell by $J$.
\end{enumerate}
The awkwardness arises from the fact that, on the Maxwell-Faraday formulation, a model $(F,J)$ just \emph{is} the theory's description of charged matter, whereas on the gauge formulation it is not.

As a result, many philosophical analyses of the equivalence between $EM_1$ and $EM_2$ are incomplete: one must still say whether the two theories are equivalent in the presence of charged matter fields. This is not to say that philosophers have not given detailed interpretations of charged matter fields as they arise from vector bundles. For example, \citet{weatherall2016ymgr} has proposed to view the principal fibre bundle as coordinating how different charged matter fields like electrons and muons ``all respond to the same electromagnetic influences'' \citep[p.2405]{weatherall2016ymgr}.\footnote{For critical discussion of this view see \citet[\S 4.4]{menon2018phd}, \citet[p.40]{jacobs2023mfb}, and \citet{gomes2026vb}. See also \citet{MarchWeatherall2025puz} and \citet{WeatherallMarch2026nt} for a recent view on the role of the principal bundle in defining an appropriate notion of general covariance.} However, Curvature Fundamentalism is concerned with a different question, that of what the physically significant quantities of gauge electromagnetism are, and whether they can be reduced to curvature. In the next section, we will argue that this sort of Curvature Fundamentalist reduction is not possible in the complete theory of gauge electromagnetism.

\section{Beyond Fundamentalism}\label{sec:inequivalence}

Let's take stock. In Section \ref{subs:curvature-fundamentalism}, we pointed out that the simplest expressions of Curvature Fundamentalism cannot be true. Namely, the claim of Determinacy ($CF_1$), that gauge invariant quantities are determined by curvature, is false: $|\psi|$ is formally gauge invariant but not determined by $F$. A popular alternative is to fall back on Empirical Sufficiency ($CF_2$), that the `physically significant' quantities are determined by $F$, or Ontological Primacy ($CF_3$), that the physical ontology is determined by $F$. In this section, we provide a counterexample that shows both these expressions of Curvature Fundamentalism are untenable as well. It also establishes a sense in which $EM_1$ and $EM_2$ fail to be theoretically equivalent.

We will begin with an easy warm-up example that we refer to as the `kinetic momentum' of a scalar field, which cannot be so quickly dismissed as `not physically significant'. However, this simple warm-up does not impose a plausible dynamical field equation, and it does not distinguish between $EM_1$ and $EM_2$. We thus follow it with the more substantial counterexample of `relative phase'. This shows that there are physically significant quantities that satisfy a dynamical field equation and which are not determined by $F$.

\subsection{Warm-Up Example: Kinetic Momentum} Consider a model $(A,\psi)$ of the gauge formulation, where $\psi$ is any complex scalar field of charge $q$ that is non-zero in some open region. Let $D_A = d - iqA$ be the covariant derivative. The \emph{kinetic momentum} of the field in this region is a local quantity given by $\mu = \Imag\left(\bar\psi D_A\psi/|\psi|^2\right)$. Writing $\sigma := |\psi|$ and $\psi=\sigma e^{i\theta}$, it can always be expressed locally\footnote{Let $D_A\psi=e^{i\theta}\bigl(d\sigma+i\sigma(d\theta-qA)\bigr)$. Then $\bar\psi D_A\psi/|\psi|^2=d\sigma/\sigma+i(d\theta-qA)$, whose imaginary part is $d\theta-qA$. This quantity appears in the Higgs mechanism \citep[p.1159]{higgs1966s}, and its significance is discussed by \citet[\S 4.2]{Struyve2011}, \citet[p.10]{wallace2014ab}, and \citet{francois2019gw}.}
\begin{equation}\label{eq:mu-def}
  \mu=d\theta-qA. 
\end{equation}
The physical significance of $\mu$ as `momentum' arises from the fact\footnote{For a non-relativistic particle with minimally coupled Hamiltonian $h(Q,P)=(2m)^{-1}|P-qA|^2+q\Phi$, Hamilton's equation gives $v=m^{-1}(P-qA)^\sharp$. In Hamilton--Jacobi theory $P=d\theta$, so $\mu=d\theta-qA$.} that the kinetic energy for a minimally coupled Hamiltonian in the Hamilton--Jacobi (or eikonal) regime is given by $\mu^2/2m$, as one would expect from a quantity $\mu$ representing momentum. It is also manifestly gauge invariant: since $\theta\mapsto\theta+q\chi$ and $A\mapsto A+d\chi$, we get that $\mu\mapsto\mu$ is left unchanged.

Despite its gauge invariance, the quantity $\mu$ is not determined by curvature, since Equation \eqref{eq:mu-def} immediately implies that
\begin{equation}
  F = dA = -d\mu/q.
\end{equation}
So, although $F$ determines the `rate of change' of $\mu$ in the sense of determining its exterior derivative, it is not enough to determine $\mu$ itself.\footnote{The latter can be checked explicitly by keeping the same $A$ and $\sigma$ fixed, and replacing $\psi=\sigma e^{i\theta}$ by $\psi'=\sigma e^{i(\theta+f)}$ for any $f$ such that $df\neq0$. These two configurations have the same $F$, but $\mu'=\mu+df \neq \mu$.} As one might expect, kinetic momentum can be used to generate a variety of further gauge invariant quantities that are not determined by curvature.\footnote{For example, the gauge invariant \emph{kinetic scalar} can be defined using a spacetime metric $g_{ab}$ on $M$ as $\mathcal K[A,\psi] := g^{ab}\overline{D_a\psi}D_b\psi = g^{ab}(\partial_a\sigma)(\partial_b\sigma)+\sigma^2g^{ab}\mu_a\mu_b$. It is gauge invariant, depends non-trivially on $A$ through minimal coupling, and is not determined by $F$.}

The physical significance of kinetic momentum $\mu$ is perhaps harder to dismiss than $|\psi|$. However, the Curvature Fundamentalist may remain sceptical. Our example here is `kinematic' in the sense that we have not provided any explicit field equation that explains how the model $(A,\psi)$ evolves dynamically. Moreover, when equipped with a minimally coupled dynamics, the four-current $J$ turns out\footnote{The minimally coupled Klein--Gordon Lagrangian is $\mathcal L_\psi=-\frac12\overline{D_a\psi}D^a\psi-\frac12m^2|\psi|^2$. Since $\delta_A(D_a\psi)=-iq\psi\delta A_a$, this implies the Noether current associated with phase symmetry is $\delta_A\mathcal L_\psi=q\Imag(\bar\psi D^a\psi)\delta A_a$. The four-current is then given by $J = q\Imag(\bar\psi D^a\psi) = q\sigma^2\mu$.} to be proportional to $\mu$,
\begin{equation}\label{eq:matter-current}
  J=q\sigma^2\mu.
\end{equation}
Curvature already determines $J$ via Maxwell's equations. So, if some further fact about the dynamics were to guarantee that $F$ determines $\sigma^2=|\psi|^2$ as well, then $F$ actually \emph{would} determine the quantity $\mu$ whenever $\sigma\neq 0$, by the equation $\mu = J/q\sigma^2$.

However, the Curvature Fundamentalist cannot always recover the missing information in this way. Our next example will show that, even when we restrict attention to models that satisfy a field equation, there are \emph{still} gauge invariant quantities that cannot be determined by $F$ alone.

\subsection{Main Example: Relative phase}\label{subs:relative-phase}

Let $(A,\psi)$ again be the field representatives for a given section, where $\psi$ is a smooth scalar field that is non-zero in some region. The base manifold $M$ may or may not be contractible. The \emph{unnormalised relative phase} along a curve $\lambda:[0,1]\to M$ with endpoints $a=\lambda(0)$ and $b=\lambda(1)$ is defined by
\begin{equation}\label{eq:unnormalised-relative-phase}
  \widetilde{\Theta}_{A,\psi}(\lambda)
  :=\bar\psi(b)\exp\left(iq\int_\lambda A\right)\psi(a).
\end{equation}
This quantity looks formally similar to the Aharonov-Bohm quantity $\theta_\lambda$ of Equation \eqref{eq:ab-relative-phase}. But, it lacks the requirement that $\lambda$ is a loop, and instead includes a charged matter field $\psi$, as shown in Figure \ref{fig:lambda-curve}. Nevertheless, it is gauge invariant.\footnote{Under $A\mapsto A+d\chi$ and $\psi\mapsto e^{iq\chi}\psi$, the endpoint factors contribute $e^{iq\chi(a)}e^{-iq\chi(b)}$, while the parallel-transport factor contributes $e^{iq\chi(b)}e^{-iq\chi(a)}$. They cancel.}

\begin{figure}[tbh]
  \centering
  \begin{tikzpicture}[
      x={0.2\textwidth},
      y={0.2\textwidth},
    ]
    \draw[very thick]
    (0.5,0.5)
    .. controls (0.67,0.5) and (0.83,0.72) ..
    (1.0,0.72)
    .. controls (1.25,0.72) and (1.75,0.28) ..
    (2.0,0.28)
    .. controls (2.17,0.28) and (2.33,0.5) ..
    (2.5,0.5);
    \node[above=4pt] at (1.5,0.5) {$\lambda$};

    \fill (0.5,0.5) circle (1.6pt);
    \fill (2.5,0.5) circle (1.6pt);
    \node[above left=2pt] at (0.5,0.5) {$a$};
    \node[above right=2pt] at (2.5,0.5) {$b$};
    \node at (0.5,0.3) {$\psi(a)$};
    \node at (2.5,0.3) {$\psi(b)$};
    
  \end{tikzpicture}
  \caption{The relative phase of a model $(A,\psi)$ of gauge electromagnetism is defined with respect to a curve $\lambda$ with endpoints $a$ and $b$.}\label{fig:lambda-curve}
\end{figure}
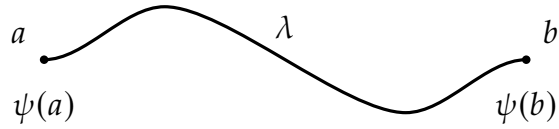

It also still contains information about the amplitude of $\psi$ at the two endpoints $a$ and $b$. To recover pure phase information, which we will shortly use to construct an interference experiment, we normalise to produce what we call the \emph{relative phase}, given by 
\begin{equation}\label{eq:relative-phase}
  \Theta_{A,\psi}(\lambda) := \frac{\widetilde{\Theta}_{A,\psi}(\lambda)}{|\psi(b)||\psi(a)|}.
\end{equation}
Relative phase is $U(1)$-valued whenever $\psi(a)$ and $\psi(b)$ are non-zero,\footnote{Since $\big|\exp\left(iq\int_\lambda A\right)\big|=1$, we have $|\widetilde{\Theta}_{A,\psi}(\lambda)| = |\psi(b)||\psi(a)|$, and hence that $|\Theta_{A,\psi}(\lambda)|=1$ as well.} in that $\Theta_{A,\psi}(\lambda)=e^{i\theta}$ for some real parameter $\theta$. Like its unnormalised counterpart, it is a gauge invariant quantity. Relative phase is also not determined by $F$ alone: let $A=0$, in which case $F=0$ as well. Writing $\psi(x) = |\psi(x)|e^{i\theta(x)}$, it follows that
\begin{equation}
  \Theta_{0,\psi}(\lambda) = \frac{\bar\psi(b)\psi(a)}{|\psi(b)||\psi(a)|} = e^{i\big(\theta(a) - \theta(b)\big)}.
\end{equation}
So, relative phase measures the differences in phase at the endpoints $a$ and $b$ of the curve, even when $A$ and $F$ both vanish. This provides another counterexample to the claim that every gauge invariant quantity is determined by $F$. In the remainder of this section, we will argue that it also has substantial physical significance, including for dynamically evolving fields associated with concrete physical experiments. This makes it a natural part of the ontology of classical electromagnetism, and motivates our conclusion that Curvature Fundamentalism is untenable. 

Relative phase is non-local, in that it depends on the values of the matter and potential fields at more than one spacetime point, which are determined by the path $\lambda$. There is a substantial mathematics and physics literature on path-dependent quantities like this. For example, \citet{giles1981wilson} reconstructs gauge potentials from Wilson loops; \citet{barrett1991holonomy} reconstructs bundles with connection from holonomy data; and \citet{schreiberwaldorf2009transport} characterise connections as transport functors on the thin path groupoid, to name a few. They have also been widely discussed as evidence of non-separability in classical electromagnetism. For example, the non-local quantity of \citeauthor{aharonovbohm1959a} appears centrally in the `holonomy interpretation' of \citet{belot1998em}, \citet{healey2007book}, and \citet{myrvold2011n}. Relative phase itself appears in the `relational' interpretation of \citet{rovelli2014gauge}. \citet{francois2019gw} associates non-local quantities with the presence of `substantial' (as compared to `artificial') gauge invariant quantities. And, such quantities feature in the  `holism' view of gauge \citep{gomes2021h,gomes2021phd}. In contrast, the recent literature on Curvature Fundamentalism has focused on point-local quantities like tensor or spinor fields and their derivatives on spacetime.\footnote{See our definition of `point-local' quantities in footnote \ref{fn:point-local}.}

What we would now like to point out is that, even when a model $(A,\psi)$ is constrained by a field equation on a contractible base manifold, there are systems for which,
\begin{enumerate}
  \item different models $(A,\psi)$ and $(A',\psi')$ of gauge electromagnetism can evolve dynamically from initially prepared field packets to a final measured interference pattern predicted by relative phase; and,
  \item  these models have vanishing curvature and four-current throughout the evolution, as well as an initial Cauchy surface on which all point-local gauge invariant quantities agree, which implies that no such quantities can predict the interference pattern.
\end{enumerate}
We take the existence of such a measurement procedure to establish both the physical and ontological significance of relative phase, contradicting Curvature Fundamentalism in all its forms. The inability of $(F,J)$ (or any point-local gauge invariant quantity) to predict this phenomenon then shows a sense in which $EM_1$ and $EM_2$ are not equivalent---not even on a contractible manifold. 

Not just any model of gauge electromagnetism will establish our claims. The interesting case arises when the endpoints of the curve $\lambda$ lie in two separated `packets' of charged matter with an empty region between them, as illustrated in Figure \ref{fig:two-packets}. The phase of each packet can then be varied without changing any gauge invariant point-local facts about either packet, giving rise to a measurable difference in relative phase without any difference in $F$ or $J$.

\begin{figure}[tbh]
\centering
\begin{tikzpicture}[
    plane/.style={line width=0.7pt},
    wave/.style={line width=0.45pt},
    every node/.style={font=\small}
]
\begin{scope}[yshift=0cm]
  \begin{scope}
    \clip (0,0) -- (9,0) -- (7.4,2.2) -- (1.6,2.2) -- cycle;
    \foreach \r in {0.15,0.35,0.55,0.75,0.95,1.15,1.35} {
      \draw[wave]
        (3.25,1.15)
        ellipse [x radius=\r, y radius={0.34*\r}];
    }
    \foreach \r in {0.15,0.35,0.55,0.75,0.95,1.15,1.35} {
      \draw[wave]
        (6.25,1.15)
        ellipse [x radius=\r, y radius={0.34*\r}];
    }
  \end{scope}
  \draw[plane]
    (0,0) -- (9,0) -- (7.4,2.2) -- (1.6,2.2) -- cycle;
\end{scope}

\begin{scope}[yshift=3.2cm]
  \begin{scope}
    \clip (0,0) -- (9,0) -- (7.4,2.2) -- (1.6,2.2) -- cycle;
    \foreach \r in {
      0.15,0.35,0.55,0.75,0.95,1.15,1.35,1.55,
      1.75,1.95,2.15,2.35,2.55,2.75,2.95,3.15,3.35,3.55,3.75,3.95,4.15
    } {
      \draw[wave]
        (3.25,1.15)
        ellipse [x radius=\r, y radius={0.34*\r}];
    }
    \foreach \r in {
      0.15,0.35,0.55,0.75,0.95,1.15,1.35,1.55,
      1.75,1.95,2.15,2.35,2.55,2.75,2.95,3.15,3.35,3.55,3.75,3.95,4.15
    } {
      \draw[wave]
        (6.25,1.15)
        ellipse [x radius=\r, y radius={0.34*\r}];
    }
    \node[fill=white, inner sep=1.5pt]
      at (4.75,1.18) {$R$};
  \end{scope}
  \draw[plane]
    (0,0) -- (9,0) -- (7.4,2.2) -- (1.6,2.2) -- cycle;
\end{scope}

\draw[->,line width=0.9pt]
  (9.35,0.2)
  .. controls (9.9,1) and (9.9,2.15) ..
  (9.35,2.95);
  \node[right] at (9.8,1.58) {$t$};

\node[above=2pt] at (4.7,5.3) {$h_1h_2\neq0$};
\node[above=2pt] at (3.25,1.5) {$h_1$};
\node[above=2pt] at (6.25,1.5) {$h_2$};

\end{tikzpicture}
\caption{Two initially separated scalar-field packets $h_1$ and $h_2$ propagate into a
region of overlap $R$ in which $h_1h_2\neq0$.}
\label{fig:two-packets}
\end{figure}
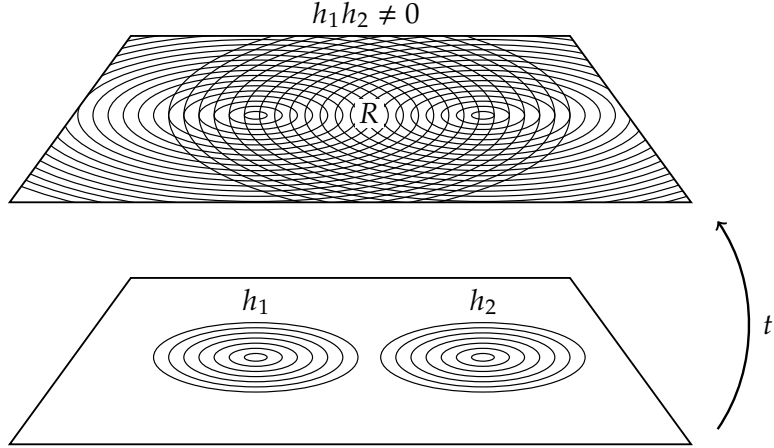

The idea is the following. Consider two real-valued solutions to the Klein-Gordon equation $h_1$ and $h_2$ on $M$, coupled in the usual way to the Maxwell equations. Suppose that on some initial Cauchy surface, the supports of their Cauchy data are disjoint, but that their corresponding solutions later overlap in a region $R$ where $h_1h_2\neq0$, and where we imagine a detector has been placed that is sensitive to these fields. The situation is shown in Figure \ref{fig:two-packets}. For any real number $c$ we can represent the combined field for these two packets by,
\begin{equation}
  \psi_c=h_1+e^{ic}h_2
\end{equation}
where $c$ parametrises the relative phase between the two packets in a given gauge. The gauge invariant scalar $|\psi_c|^2$ of the combined field is given by
\begin{equation}\label{eq:gauge-invariant-scalar}
  |\psi_c|^2 = h_1^2 + h_2^2 + h_2h_1\left(e^{ic} + e^{-ic}\right) = h_1^2 + h_2^2 + 2h_2h_1\cos c.
\end{equation}
For an idealised detector whose response is proportional to $|\psi|^2$, the final term $f(c) = 2h_2h_1\cos c$ describes the detection of interference between the two packets, since it is the cross-term that remains after subtracting the responses due to the individual packet contributions $h_1^2$ and $h_2^2$. As expected, interference is zero on the initial Cauchy surface where $h_1h_2=0$. But, it may be non-zero in the future region where they evolve to overlap with $h_1h_2\neq0$.

Let us now adopt the simple potential $A=0$, and consider a curve $\lambda$ in the initial Cauchy surface, with one endpoint $a$ in the first packet satisfying $h_1(a)>0$ and the other endpoint $b$ in the second packet satisfying $h_2(b)>0$. Then, at each point in the future region $R$, the interference $f(c)$ is proportional to the real part of relative phase: since $\Theta_{0,\psi_c} = e^{-ic} = \cos c - i\sin c$, we get that $\Re\Theta_{0,\psi_c} = \cos c$, and hence the interference term is given by,
\begin{equation}
  f_c = 2 h_2h_1\cos c = 2h_2h_1\Re\Theta_{0,\psi_c}.
\end{equation}

We will restrict attention to the cases in which $c=0$ or $c=\pi$. First, imposing the condition that $A=0$ now places further constraints on $\psi_c$ so as to ensure the Maxwell equations are satisfied. This is because the four-current $J_c$ for our field $\psi_c$ is given\footnote{Since $A=0$, we have $D_A=d$, and therefore $\bar\psi_cD_A\psi_c=h_1dh_1+h_2dh_2+e^{ic}h_1dh_2+e^{-ic}h_2dh_1$. The first two terms are real, so taking the imaginary part gives $J_c=q\sin(c)(h_1dh_2-h_2dh_1)$.} by $J_c=q\sin(c)(h_1dh_2-h_2dh_1)$. But, since we have set $A=0$, we get that $F=0$ too, and so the Maxwell equations require that $J=0$ as well. We ensure this by considering only cases in which $c=0$ or $c=\pi$, since then $\sin c=0$ and so the four-current $J_c$ vanishes on the entire spacetime. However, since $\cos 0 = 1$ and $\cos \pi = -1$, these are also cases in which there may be non-trivial interference $f_0$ and $f_\pi$ in future region $R$. These cases thus arise from different values of relative phase $\Theta=1$ and $\Theta=-1$ on the initial data surface. In other words, the two initial data descriptions $c=0$ and $c=\pi$ dynamically evolve into two measurably distinct interference patterns described by the system's relative phase, whereas the curvature and four-current remain zero throughout each process. As a result, this interference effect cannot be reduced to facts about curvature and four-current. Our observation can be summarised in the following (a proof is given in the Appendix):

\begin{restatable}{proposition}{relationalprop}\label{prop:relational} 
Let $M=\RR^4$ be equipped with the Minkowski metric, and let $h_1$ and $h_2$ be complete real-valued solutions to the free Klein-Gordon equation with smooth, compactly supported Cauchy data whose supports are disjoint on an initial Cauchy surface. Suppose there are points $a$ and $b$ in their respective supports such that $h_1(a)>0$ and $h_2(b)>0$. Let $h_1h_2\neq0$ throughout some future region $R$. Finally, let
\begin{equation}\label{eq:relative-phase-pair}
  \psi_0:=h_1+h_2, \qquad \psi_\pi:=h_1-h_2.
\end{equation}
Then the two models $(A,\psi) = (0,\psi_0)$ and $(A',\psi') = (0,\psi_\pi)$ satisfy the following:
\begin{enumerate}
  \item[(a)] each has curvature and four-current given by $(F,J)=(0,0)$ everywhere;
  \item[(b)] each is a complete solution to the coupled Maxwell-Klein-Gordon equations;
  \item[(c)] the two models agree on all point-local gauge invariant quantities on the initial surface;
  \item[(d)] the two models are not related by a gauge transformation;
  \item[(e)] the two models display different relative phase with respect to a curve $\lambda$ with endpoints $a,b$, and display different interference in $R$.
\end{enumerate}
\end{restatable}

This result provides a straightforward measurement protocol for detecting the difference in these models empirically:
\begin{enumerate}
  \item Prepare two separated packets;
  \item allow them to propagate into an overlap region; and
  \item measure the resulting degree of interference in that region.
\end{enumerate}
Indeed, a familiar analogue is the propagation of two classical waves from different sources with a phase shifter, which later overlap and create an interference pattern on a detector screen. Beam-splitter experiments of this kind were popularised by \citet{tHooft1980g} and appear in discussions of the empirical significance of gauge theory given by \citet{bradingbrown2004gauge}, \citet{Struyve2011}, \citet{greaveswallace2014symmetries}, \citet{teh2016gg}, and \citet{MurgueitioRamirez2022p}; see also \citet{gomes2021h,gomes2025el}.

We have thus established that, even on a contractible manifold, there is a gauge invariant quantity that cannot be determined by curvature, and which describes the outcome of a concrete interference experiment resulting from a dynamically evolving field. Since two fields can be chosen so as to produce different experimental outcomes for the same Maxwell models $(F,J)$, we can see no plausible sense in which the `physically significant' or `ontologically relevant' quantities in gauge electromagnetism are reducible to $(F,J)$ alone.

\subsection{Inequivalence of $EM_1$ and $EM_2$}\label{sec:consequences-for-equivalence}

Some of the philosophical consequences that have been drawn from Curvature Fundamentalism must now be revisited. For example, \citet{weatherall2016ng,weatherall2016ug} has used Curvature Fundamentalism to argue that $EM_1$ and $EM_2$ are equivalent:
\begin{quote}
    ``We stipulate that on both formulations, the empirical content of a model is exhausted by its associated Faraday tensor. In this sense, the theories are empirically equivalent, since for any model of $EM_1$, there is a corresponding model of $EM_2$ with the same empirical content (for some fixed $J^a$), and vice versa.'' \citep[p.1078]{weatherall2016ng}
\end{quote}
Many philosophers agree. For example, \citet[p.5]{wolfread2023} cite a ``near-universal consensus'' that the formulations are theoretically equivalent on a contractible manifold, with similar remarks recently made by \citet[p.1258]{moller-nielsen2017i}, \citet[p.677]{nguyentehwells2020g}, and \citet[Proposition 2]{chen2024s}.

This claim can only be reasonably entertained in the special case that we called an `awkward' comparison in Section \ref{subs:curvature-fundamentalism}. Suppose one restricts attention to contractible base manifolds, and to the sector of electromagnetism in which there are no charged matter fields unless they are determined by $J$. Then it is not implausible to view $EM_1$ and $EM_2$ as equivalent. But, once dynamical matter is included, not only does Curvature Fundamentalism fail, but the equivalence between $EM_1$ and $EM_2$ fails as well. The two models $(A_1,\psi_1) = (0,\psi_0)$ and $(A_2,\psi_2) = (0,\psi_\pi)$ of $EM_2$ constructed above are empirically distinguishable and gauge inequivalent, but correspond to the very same model $(F,J)$ of $EM_1$, even on a contractible base manifold. As a result, these theories are not equivalent, in that there are measurable experimental results predicted by one but not the other.

Our result has further implications for the broader reductionist programme: we have shown that \emph{no} gauge invariant point-local quantities constructed from $(A,\psi)$ can predict the interference experiment above: in our two models, all gauge invariant point-local quantities agree on the initial data surface, despite that surface evolving into two different interference patterns in the future. As a result, \emph{no reduction to gauge invariant point-local (e.g. tensorial) quantities constructed from $(A,\psi)$ on an initial Cauchy surface can adequately capture the empirical predictions of gauge electromagnetism.} A complete reduction must instead include cross-region, path-dependent information not defined at a point, either by including quantities like relative phase or through some other equivalent structure. That said, it still may be possible to formulate such an enriched reduced theory entirely in gauge-invariant terms. We also view the rich structure of gauge-invariant quantities like relative phase as good reason to be sophisticated about symmetries as well, in the sense of \citet{dewar2019s}. However, with charged matter included, there is no longer an equivalence argument against sophistication of the kind suggested by \citet[p.14]{BradleyWeatherall2026s} in the matter-free sector.

\section{Conclusion}

Two of the founders of gauge electromagnetism drew the following conclusion:
\begin{quote}
``(a) The field strength $f_{\mu\nu}$ underdescribes electromagnetism, i.e. different physical situations in a region may have the same $f_{\mu\nu}$. (b) The phase (1) overdescribes electromagnetism, i.e., different phases in a region may describe the same physical situation.'' \citep[p.3856]{wuyang1975g}.
\end{quote}
Our conclusion is similar. On a contractible region, curvature provides a complete description of the matter-field-free electromagnetic connection, in that it determines the potential up to gauge transformation. However, if charged matter is included, or if contractibility is relaxed, then it does not. There are quantities like relative phase whose dynamics are not entirely visible to curvature and four-current. In this more complete context, Curvature Fundamentalism fails in all its forms, and the equivalence between $EM_1$ and $EM_2$ comes apart. 

One of the most striking philosophical consequences of our analysis is for the general project of reduction. This is the idea that one ought to be able to reduce `gauge quantities' to `spacetime quantities', the latter of which are typically formulated using fields defined point-locally like tensors and spinors. However, relative phase is defined across multiple points. Our result shows that its predictions cannot be captured by any reduction to gauge invariant point-local quantities constructed from \((A,\psi)\) on an initial Cauchy surface. We do not see this as a reductio of the reductionist project. However, our analysis does lead to a question that has not been widely appreciated by philosophers: what sort of quantities is a reductionist formulation of electromagnetism about? The status of this question appears to remain open, and worthy of the reductionist's pursuit.

\section*{Appendix}

\relationalprop*
\begin{proof}
  (a) From $A=0$ we have that $F=dA=0$. Since $\psi$ is real, it also follows that $J=0$, since $\bar\psi D_A\psi=\psi d\psi$ is real, so $J=q\Imag(\bar\psi D_A\psi)=0$ as well.

  (b) We have just shown the Maxwell equations are satisfied with $(F,J)=(0,0)$. That $\psi_0$ and $\psi_\pi$ are also complete real-valued solutions to the free Klein-Gordon equation follows from our assumption that $h_1$ and $h_2$ are complete solutions, together with the fact that the Klein-Gordon equation is linear when $A=0$. 

  (c) Since the supports of the Cauchy data are disjoint on the initial Cauchy surface, every point on that surface admits a neighbourhood on which the Cauchy data of at most one packet are nonzero. On a neighbourhood of the first packet the two data sets coincide; on a neighbourhood of the second they differ by a factor of $-1$; outside both they vanish. As a result, all tangential finite jets are locally related by a constant gauge transformation. But, the Klein--Gordon equation recursively determines higher normal derivatives from the Cauchy data \citep[p.246]{wald1984gr}, and so the full finite spacetime jets are related by the same constant gauge transformation. Since the electromagnetic data agree as well, every point-local gauge-invariant function of these jets therefore takes the same value in the two models on the initial surface.

  (d) Suppose for reductio that the two models are related by a gauge transformation, with $A\mapsto A+d\chi$ and $\psi_\pi = e^{iq\chi}\psi_0$. The former preserves $A=0$, so $d\chi=0$. On a connected base manifold this implies $\chi$ is constant. But, $\chi$ cannot be constant, because at $a$, where $h_1(a)>0$, we have $\psi_0 = h_1 = \psi_\pi$ with $e^{iq\chi}=1$, while at $b$, where $h_2(b)>0$, we have $\psi_0 = h_2 = -\psi_\pi$ with $e^{iq\chi}=-1$. Thus, the two models cannot be gauge equivalent.

  (e) By explicit calculation, the relative phase is $\Theta_{0,\psi_0}(\lambda)=1$ for $\psi_0$ and $\Theta_{0,\psi_\pi}(\lambda)=-1$ for $\psi_\pi$. The interference of each model is given by $f_c = 2h_1h_2\Re \Theta_{0,\psi_c}$. Thus,
  \begin{equation}
    f_0 - f_\pi = 2h_1h_2\Re(\Theta_{0,\psi_0} - \Theta_{0,\psi_\pi}) = 4h_1h_2,
  \end{equation}
which by hypothesis is non-zero in $R$. Thus, these two models display different interference.
\end{proof}

\setstretch{1.0}

\end{document}